\documentclass[10pt,journal]{IEEEtran}
\usepackage{cite}	
\usepackage{graphicx} 
\usepackage[latin1]{inputenc} 
\usepackage[T1]{fontenc} 
\usepackage{amsmath,amsfonts,amsbsy,amssymb} 
\usepackage{mathabx} 
\usepackage{amssymb} 
\usepackage{amsmath}
\usepackage{amsthm}	
\usepackage{mathrsfs}
\usepackage[nolist]{acronym} 
\usepackage{tabularx} 
\usepackage{multirow}
\usepackage{wasysym}
\usepackage{float}
\usepackage{color} 

\usepackage{enumitem} 

\newtheorem{lemma}{Lemma}
\newtheorem{theorem}{Theorem}

\newtheorem{remark}{Remark}
\newtheorem{problem}{Problem}

\begin{document}

\title{Effective Capacity and Power Allocation for Machine-Type Communication}
\author{
	\IEEEauthorblockN{Mohammad Shehab, Hirley Alves, and Matti Latva-aho}
	\thanks{		
	M. Shehab, H. Alves, and M. Latva-aho are with Centre for Wireless Communications (CWC), University of Oulu, Finland
		Email: firstname.lastname@oulu.fi} 
	\thanks{This work is partially supported by Academy of Finland 6Genesis Flagship (Grant no. 318927), Aka Project EE-IoT (Grant no. 319008), and by Finnish Funding Agency for Technology and Innovation (Tekes), Bittium Wireless, Keysight Technologies Finland, Kyynel, MediaTek Wireless, and Nokia Solutions and Networks.}
}
%
\maketitle

\begin{abstract}
	Effective capacity (EC) determines the maximum communication rate subject to a particular delay constraint. In this work, we analyze the EC of ultra reliable Machine Type Communication (MTC) networks operating in the finite blocklength (FB) regime. First, we present a closed form approximation for EC in quasi-static Rayleigh fading channels. Our analysis determines the upper bounds for EC and delay constraint when varying transmission power. Finally, we characterize the power-delay trade-off for fixed EC and propose an optimum power allocation scheme which exploits the asymptotic behavior of EC in the high SNR regime. The results illustrate that the proposed scheme provides significant power saving with a negligible loss in EC. 
\end{abstract}


\vspace{-0mm}
\section{Introduction}\label{introduction}
Communication systems have become everyday use equipments everywhere around us. In these systems, information is commonly conveyed in the form of data bits which are then transformed to coded packets. Packets are then transmitted in noisy mediums which are affected by fading. For a certain communication channel, Shannon capacity determines the attainable rate by which information can be transmitted with almost no error. Conventionally, communication systems are designed based on Shannon theory, which resorts to the transmission of relatively long data packages when there is a large number of channel uses per packet. Machine type communication (MTC) systems ranging from sensor to vehicular networks often have strict delay constraints, where packets are relatively short and required to be transmitted at minimum latency and a high level of reliability (i.e, >99.99$\%$). This is not merely achieved via conventional coding with long blocklength. Meanwhile, ultra reliable communication (URC) has evolved to propose solutions for reliable and low latency communication. The next generations of mobile communication are expected to support such demands via MTC \cite{paper1,paper2,TVT3}. 

To achieve minimum latency and ultra reliability as envisioned for real time applications and emerging technologies such as e-health and road safety, these networks communicate on short messages. Transmission of short packets does not comply to Shannon capacity, which becomes a poor performance metric at finite blocklength as pointed out in \cite {paper2}. Communication in the finite blocklength regime has gained an increasing attention in the recent years \cite {paper2,Yang_J, paper5,Yang2014c,TVT1,TVT2}, especially after the seminal work in \cite{Polyanskiy2010}, where coding rates of short packets are defined for the additive white Gaussian noise (AWGN) channel. Later, in \cite{Yang2014c} the authors determined the maximum communication rate as a function of blocklength and error outage probability in quasi-static fading. The results highlighted a constant gap between their achievability bound and practical coding schemes which are implemented in current standards. However, the gains in latency due to short messages come at a cost of reliability as discussed in \cite{paper2}. 

Effective capacity (EC) is defined as a measure of the highest arrival rate which can be served by the network under particular latency constraint. It is a metric which was first introduced in \cite {paper6} to capture the physical and link layers characteristics by insuring specific quality of service guarantees. Thus, it allows us to investigate further the latency-reliability trade-off. In \cite{TVT1} and \cite{TVT2}, finite blocklength performance of cooperative and relay-assisted networks was discussed but without considering power allocation or latency aspects. In our work, we resort to the EC theory to analyze the latency and data rate for finite blocklength packets. In \cite{paper5}, Gursoy discussed the statistical framework of effective capacity given in bits per channel use (bpcu) of one node in Rayleigh block fading environment where the channel coefficients are constant through one block transmission time. The EC was defined as a function of error probability and delay quality of service (QoS) exponent. However, they did not present a closed form expression for the effective capacity in their work. Musavian et al. analyzed the EC maximization of a cognitive network in \cite{paper9} and investigated the EC maximization subject to effective energy efficiency constraint in \cite{paper8}. The per-node EC in massive MTC networks was studied in \cite{eucnc} proposing three methods to alleviate interference namely power control, graceful degradation of delay constraint and the hybrid method which is based on the first two.

Herein, we analyze the EC of short packet transmission in quasi-static Rayleigh fading channels under delay exponent limit. The motivation of this work is to provide a mathematical framework for the performance analysis of ultra reliable low latency MTC. Moreover, we aim at characterizing the power delay relation and how to allocate power efficiently in the high SNR regime. The contributions of this work are summarized as follows: \textit{i)} a closed-form approximation for the EC is obtained in terms of incomplete gamma function, which facilitates the derivation of the optimum error probability that maximizes the EC, therefore the maximal transmission rate; \textit{ii)} we derive an upper bound for EC for fixed delay exponent and an upper bound of delay exponent for fixed EC; \textit{iii)} we characterize the amount of power required to support a certain EC while meeting its delay requirement; and \textit{iv)} we propose a power saving scheme which determines the optimum power allocation at the high SNR regime.

The rest of the paper is organized as follows: in Section \ref {pre}, we introduce the system model and some definitions. Next, we obtain a tight closed form approximation for the EC in quasi-static Rayleigh fading environment and characterize the optimum error probability in Section \ref{EC_FB}. In Section \ref{power_delay}, we discuss the power-delay trade off and propose an optimum power allocation scheme in the high SNR regime. The results are depicted in Section \ref{results}. Finally, Section \ref{conclusion} concludes the paper.
\vspace{-2mm}
\section{Preliminaries} \label{pre}

\subsection{System Layout}
We consider a point to point transmission in which short packets are transmitted through a quasi-static fading channel, and full channel state information (CSI) is assumed, thus allowing rate adaption as in \cite{paper5}. Also, the fading is considered to be Rayleigh distributed where the coefficients remain constant over $n$ symbols spanning the whole packet duration. Note that the fading coefficients $Z=|h|^2$ are exponentially distributed, thus $f_Z(z)=e^{-z}$. In short packet transmission, packets are conveyed at a rate that is not only a function of the SNR, but also the blocklength $n$, and the probability of error $\epsilon \in\left[ 0,1\right]$ \cite{paper2}, as illustrated in Fig.~\ref{fig:syst-model}. In this case, $\epsilon$ has a small value but not vanishing. In Fig. 1, a packet buffer is the memory space which stores packets awaiting transmission over the network. Packets are stored temporarily and then transmitted in a FIFO process. Thus, transmission delay occurs for packets while waiting in the buffer. Assume that each symbol is transmitted with SNR $\rho$ in a channel whose fading coefficients are denoted by $h$ and the noise is zero-mean AWGN. Thus, the normalized achievable rate in bpcu is approximated by\footnote{The approximation is accurate for blocklength $n\ge 100$ as demonstrated in \cite[Figs. 12 and 13]{Polyanskiy2010} for AWGN channel, and in \cite{Yang2014c} for fading channels.}	 
\begin{align}\label{eq3}
\begin{split}
r\approx C\left( \rho|h|^2\right) -\sqrt{\tfrac{V(\rho |h|^2)}{n} }  \operatorname{Q}^{-1}
(\epsilon), 	
\end{split}
\end{align}
where $C(x)=\log_2(1+x)$ is Shannon's channel capacity, $V(x)=\left(1-\left( 1+x\right)^{-2}\right)\left(  \log_2 e\right) ^2$ denotes the channel dispersion \cite{paper5}, $\operatorname{Q}(x)=\int_{x}^{\infty}\tfrac{1}{\sqrt{2 \pi}}\exp({\tfrac{-t^2}{2}}) dt$ is the Gaussian Q-function. 
\begin{figure}[!t] 
	\centering
	\includegraphics[width=1\columnwidth]{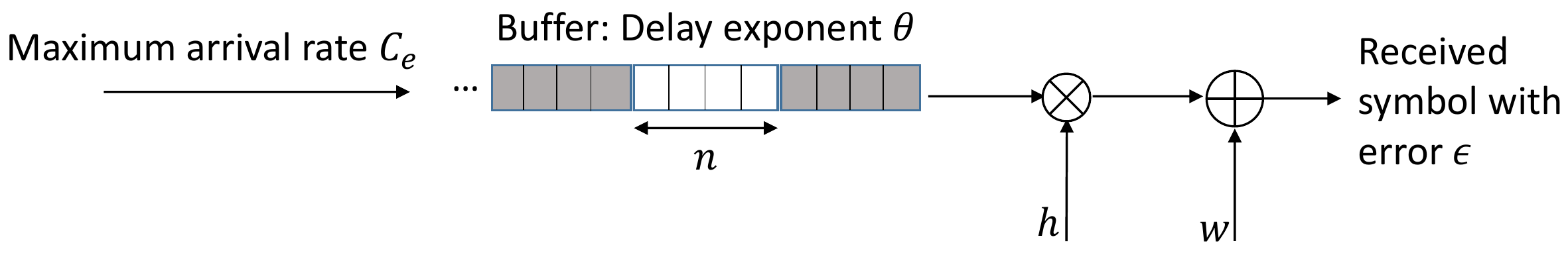}
	\vspace{-2.5mm}
	\caption{Transmission in finite blocklength with delay exponent $\theta$. One block consists of $n$ symbols.}
	\label{fig:syst-model}
	\vspace{-2mm}
\end{figure}
\subsection{Effective Capacity of finite blocklength packets} \label{EC_EEE} 
EC denoted as $C_e$ defines the maximum rate that a communication network can transfer data with, while maintaining certain delay limits in terms of delay outage probability and maximum delay bound $D_{max}$. The delay outage probability is the probability that the transmission delay exceeds the maximum delay bound $D_{max}$ channel uses and hence, an outage occurs. The delay outage probability is given by \cite{paper6} \vspace{-1mm}
\begin{align}\label{delay}
P_{out\_ delay}=\Pr(delay \geq D_{max}) \approx e^{-\theta  C_e  D_{max}},
\end{align}	
where $\Pr(\cdot)$ means the probability of the event between brackets. The delay exponent $\theta$ indicates the system's tolerance to long delays. Small values of $\theta$ mean that the network tolerates longer delays. Conversely, higher $\theta$ values mean that the system is less tolerable to longer delays. In quasi-static fading, the EC in bpcu is given by \cite{paper5}
\begin{align}\label{EC}
C_e(\rho,\theta,\epsilon)=-\frac{\ln\ \mathbb{E}_{Z}\left[\epsilon+(1-\epsilon)e^{-n\theta r}\right]}{n\theta}, 
\end{align} 
where the variable rate $r$ is given in (\ref{eq3}), and $\mathbb{E}_{Z}(\cdot)$ is the expectation of the fading distribution. 

\begin{remark}Contrary to the infinite blocklength model (as in \cite{paper8}), the EC is upper bounded in the finite blocklength regime by \vspace{-2mm}
\begin{align}\label{ub5}
\begin{split}
C_{e_b} =\lim\limits_{\rho\rightarrow \infty} \frac{-\ln\left(\mathbb{E}_{Z}\left[\epsilon+(1-\epsilon)e^{-n\theta r}\right]\right)}{n\theta} =-\frac{\ln(\epsilon) }{n \theta},
\end{split}
\end{align}
where the rate $r$ also tends to $\inf$. It follows that the EC is asymptotic at the high SNR regime towards $-\frac{\ln(\epsilon)}{n \theta}$ which is independent of power. Likewise, there is an upper bound for the delay exponent that can be supported by a buffer with constant EC which is given by 
\begin{align}\label{ub6}
\begin{split}
\theta_b=\lim\limits_{\rho\rightarrow \infty} \frac{-\ln\left(\mathbb{E}_{Z}\left[\epsilon+(1-\epsilon)e^{-n\theta r}\right]\right)}{n \ C_e} =-\frac{\ln(\epsilon) }{n \ C_e}.
\end{split}
\end{align}
\end{remark}

Note that, the bounds in (\ref{ub5}) and (\ref{ub6}) disappear as the blocklength tends to infinity and the error probability vanishes. The  EC is asymptotic in the finite blocklength regime and not monotonically increasing as in the infinite blocklength model.

\subsection{Effective Capacity in quasi-static Rayleigh fading} \label{EC_FB} 
In \cite{paper5}, a stochastic model for EC under finite blocklength coding was introduced, but numerically evaluated. Herein, we propose a tight approximation for the EC in quasi-static Rayleigh fading.
\begin{lemma} \label{lemma 1}
	For a quasi-static Rayleigh fading channel with blocklength $n$, the EC is approximated as
	\begin{align}\label{Rayleigh}
	\begin{split}
	C_e(\rho,\theta,\epsilon)\approx&-\frac{1}{n\theta} \ln \left[ \epsilon+(1-\epsilon) \ \mathcal{J}\right], 
	\end{split}
	\end{align}	
	where 
	\begin{align}\label{c2.2}
	\begin{split}
	\mathcal{J}\!=\!e^{\frac{1}{\rho}}\! \rho^\alpha  \!\left[\!\vphantom{\Gamma\left(\alpha-1,\frac{1}{\rho}\right) }\right. &  \left. \!\left(\kappa 
	\!+\! 1\right)  \!\Gamma\left(\!\alpha+1,\frac{1}{\rho} \right) \right. 
	\left.\!-\frac{\kappa-\frac{\beta}{2} }{\rho^{2}} 
	\Gamma\left(\!\alpha\!-\!1,\frac{1}{\rho}\right)\! \right],  
	\end{split}
	\end{align} 		
and $\kappa = \frac{\beta^2}{2}+\beta$, $\alpha=\frac{-\theta n}{\ln 2}$, $\beta=\theta \sqrt{n} Q^{-1}(\epsilon)\log_2e$, and $\Gamma(\cdot, \cdot)$ is the upper incomplete gamma function \cite[\S 8.350-2]{Gradshteyn}.
\end{lemma}
\begin{proof}
Please refer to Appendix A.
\end{proof}

\begin{lemma} \label{lemma 2}
	There is a unique global maximizer for the EC in the error probability $\epsilon$ in quasi-static Rayleigh fading channels which is given by 
\begin{align}\label{e*}
\begin{split}
\epsilon^*(\rho,\alpha,\beta)\approx\arg\min_{0 \leq \epsilon \leq 1} \  \epsilon+(1-\epsilon) \ \mathcal{J}.
\end{split}
\end{align} 
	%
\end{lemma}
\begin{proof}
	The expectation given by (\ref{EC2}) was proven to be convex in $\epsilon$ in \cite{paper5} for any distribution of the channel coefficients. Note that $\mathcal{J}$ is a function of $\epsilon$ as indicated in Lemma 1, in (7) and the auxiliary variable $\beta$. Therefore, the unique maximizer of the EC is the same as the minimizer of (\ref{EC2}) given by $\epsilon^*$ in (\ref{e*}). 
\end{proof}
Notice that $\epsilon^*$ from Lemma \ref{lemma 2} is attained via exhaustive search (or numerical solvers available in Matlab). Then the maximum effective capacity $C_{e_{max}}$ is obtained by plugging the solution of (\ref{e*}) into (\ref{Rayleigh}). Thus, the importance of Lemma 2 is to assure that for delay limited applications, in order to maximize the rate subject to a certain delay constraint, it is optimal to transmit with a non-zero error probability.

\section{Asymptotic analysis} \label{power_delay}
Due to the asymptotic nature of EC in the finite blocklength regime which is concluded in (\ref{ub5}), there is a slight gain in EC when increasing transmit power in the high SNR region. In this section, we aim at observing the delay bound characterized in (\ref{ub6}) for fixed EC and exploit the asymptotic behavior of EC in the high SNR regime in order to save power. For this purpose, from Lemma \ref{lemma 1} we propose a tight approximation for the power delay profile at high SNR as follows:
\begin{lemma} \label{lemma 3} At high SNR ($\rho \rightarrow \infty$), the required SNR to achieve an EC of $C_e$ for a certain delay constraint $\theta$ is approximated as
\vspace{-0mm}
\begin{align}\label{pd0}
\begin{split}
\rho(C_e,\theta,\epsilon)\approx\frac{1}{\mathcal{W}\left(\frac{\mathcal{G}}{\mathcal{F}} \right)},
\end{split}
\end{align}
where $\mathcal{G}=\frac{e^{-n \theta C_e}-\epsilon}{1-\epsilon}$, $\mathcal{F}=\frac{\kappa-\frac{\beta}{2}}{\alpha-1}-\frac{\kappa+1}{\alpha+1}$, and $\mathcal{W}(.)$ is the Lambert-W function \cite{Corless}.

\end{lemma}
\begin{proof} Please refer to Appendix B.
\end{proof}

Note that (\ref{pd0}) gives the amount of power needed for a certain buffer to support a fixed EC in the high SNR regime. In order to reduce power consumption and guarantee negligibly degraded EC, we maximize the EC subject to a constraint on the rate of change of EC with respect to SNR. This can be formulated as the following minimization problem with a maximum power constraint $\rho_{max}$:
\begin{problem} \label{problem 1}

\begin{subequations}
	\begin{align}
	\min_{\rho \geq 0} \ &\psi(\rho,\theta,\epsilon) =\epsilon+(1-\epsilon) \mathcal{J}, \\
	\mathrm{s.t} \ &\rho \leq \rho_{max}, \label{b} \\
	&\frac{\partial C_e}{\partial \rho}\geq \mu. \label{c}
	\end{align}
\end{subequations}	
\end{problem}
Note that $\mu$ is the power saving factor and the minimum acceptable rate of change of the EC with respect to power. The constraint (10c) guarantees that if transmit power is raised, there will be a reasonable gain in EC. Moreover, $\mu$ is chosen to be a positive number close to zero so that the constraint is not extremely tight and the loss in EC is negligible. For instance, for a noise power of 1 mW, $\mu=10^{-2}$ means that doubling the transmit power results in only $10^{-2}$ bpcu increase in the EC.
\begin{theorem} \label{theorem 2}
At high SNR, the solution of Problem \ref{problem 1} admits an optimum power allocation policy given by
\begin{align}\label{pd4}
\begin{split}
\rho^*=\min \left\lbrace \frac{2}{\sqrt{1+4\sqrt{\frac{\mu n \theta \epsilon}{ (1-\epsilon)\mathcal{F}}}}-1},\rho_{max}\right\rbrace, 
\end{split}
\end{align}
where $\mathcal{F}$ is defined in Lemma \ref{lemma 3}.
\end{theorem}
\begin{proof} Please refer to Appendix C.
\end{proof}
The saved power in dB is given by 
\vspace{-2mm}
\begin{align}\label{pd11}
\begin{split}
\eta=\rho_{max}-\rho^*.
\end{split}
\end{align}
\section{Numerical Results} \label{results}
In Fig. \ref{new_formula}, we present the EC in quasi-static Rayleigh fading channel when varying delay exponents. The plots apply the expectation in (\ref{EC}) and Lemma \ref{lemma 1}. The system parameters are $n=500$ channel uses, $\epsilon=10^{-4}$. The figure corroborates the accuracy of Lemma \ref{lemma 1} specially in the high SNR regime of interest with and the error of $0.001\%$ where the error is defined as $100\times|(a-b)|/a$, where $a$ is given in \eqref{EC} and b denotes \eqref{Rayleigh}. Furthermore, the figure shows the upper bound of EC in the finite blocklength regime obtained from (\ref{ub5}) and the unbounded EC when applying Shannon's model which assumes infinite blocklength and zero error. This result is expected since in the finite blocklength regime, the rate is not only bounded to the SNR but also to the error probability. However, the performance gap is only significant in the high SNR regime. That is why we concentrated our analysis in Section \ref{power_delay} on the high SNR region while \cite{paper8} provides an alternative for other transmission regions relying on conventional coding.

Fig. \ref{power_delay_profile} depicts the power delay profile for different QoS constraints at high SNR operating with fixed EC. We obtain these plots by applying Lemma \ref{lemma 3}. The figure shows that for a fixed EC value, the consumed power grows exponentially when the delay exponent becomes more strict. Furthermore, we observe the upper bound of the delay exponent that can be supported for each fixed EC value obtained from (\ref{ub6}) which also validates the tightness of Lemma 3. Thus, in order to support higher delay constraints, EC should be suppressed. This implies low rate transmission with ultra low latency which serves the intuition of MTC.
\vspace{-0mm}

In Fig. \ref{power_gain}, we plot the power gain obtained from applying Theorem 1 for $n=500$, $\rho_{max}=20$ and $30$ dB and $\epsilon=10^{-4}$. Note that the red curve represents the no gain line ($\mu=0, \eta=0$ dB). We observe that higher power gains can be achieved when the power saving factor $\mu$ and the delay constraint $\theta$ increase. This occurs due to the fact that rising the transmit power renders limited gain in EC for delay strict networks as pointed in Proposition 1 in \cite{eucnc}. Notice that as the power saving parameter $\mu$ becomes higher, we obtain higher power gain. However, this comes at the cost of greater loss in the EC as pointed out in Fig. \ref{EC_difference}, which shows the EC as function of the delay exponent. 

Fig. \ref{EC_difference} depicts the EC loss due to the power saving obtained from our proposed power allocation strategy. It is obvious that the asymptotic nature of the EC in the finite blocklength regime provides considerable power saving at a very low loss in EC. This can be noticed specially at the high SNR regime where the simulation is performed.

\begin{figure}[t] 
	\begin{center}
		\includegraphics*[width=1\columnwidth]{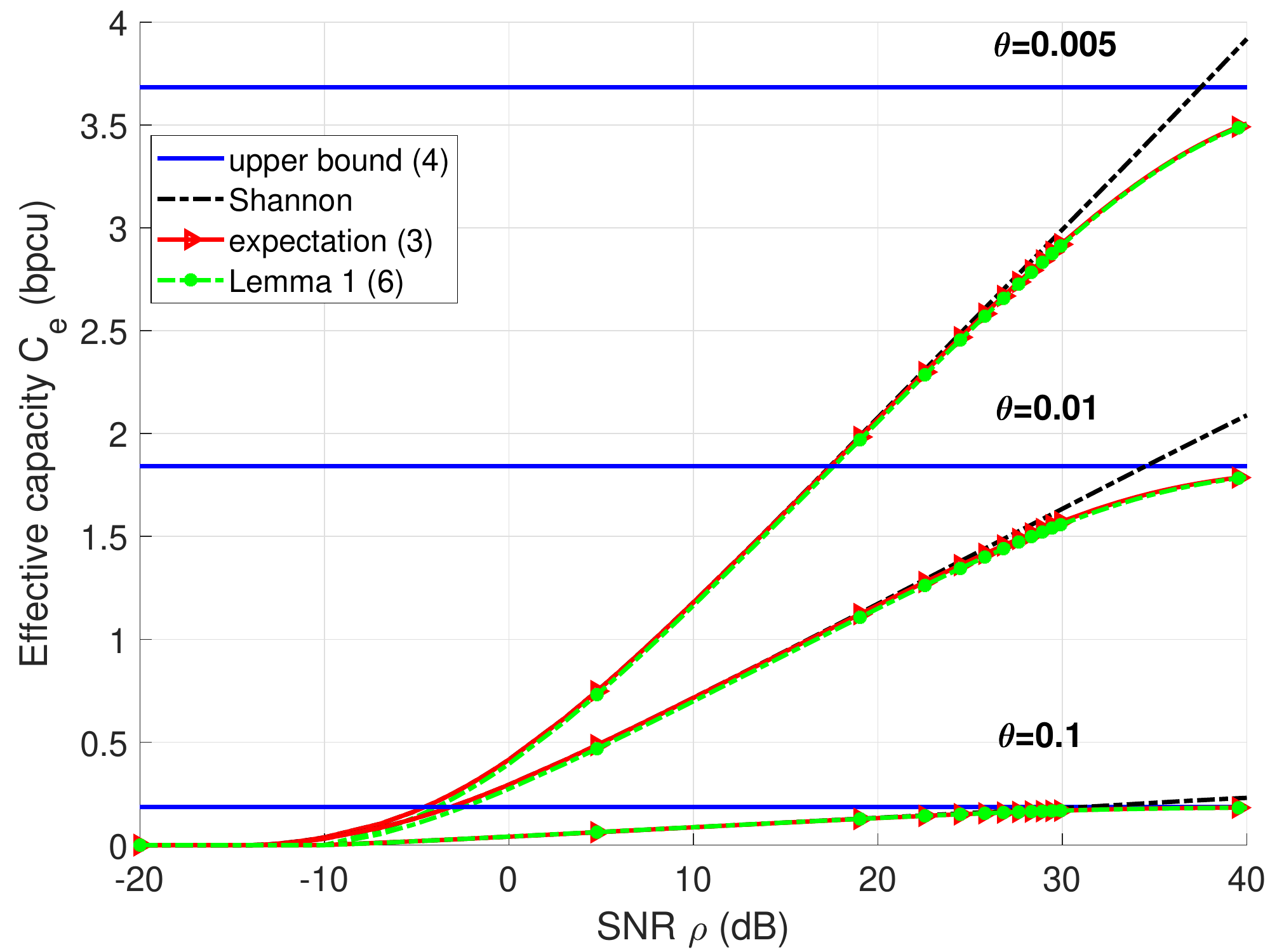}
	\end{center}
	\vspace{-0mm}
	\caption{Effective capacity as a function of SNR in quasi-static Rayleigh fading for $n=500, \epsilon=10^{-4}$.}
	\label{new_formula}
	\vspace{-0mm}
\end{figure}

\begin{figure}[ht] 
	\begin{center}
		\includegraphics*[width=1\columnwidth]{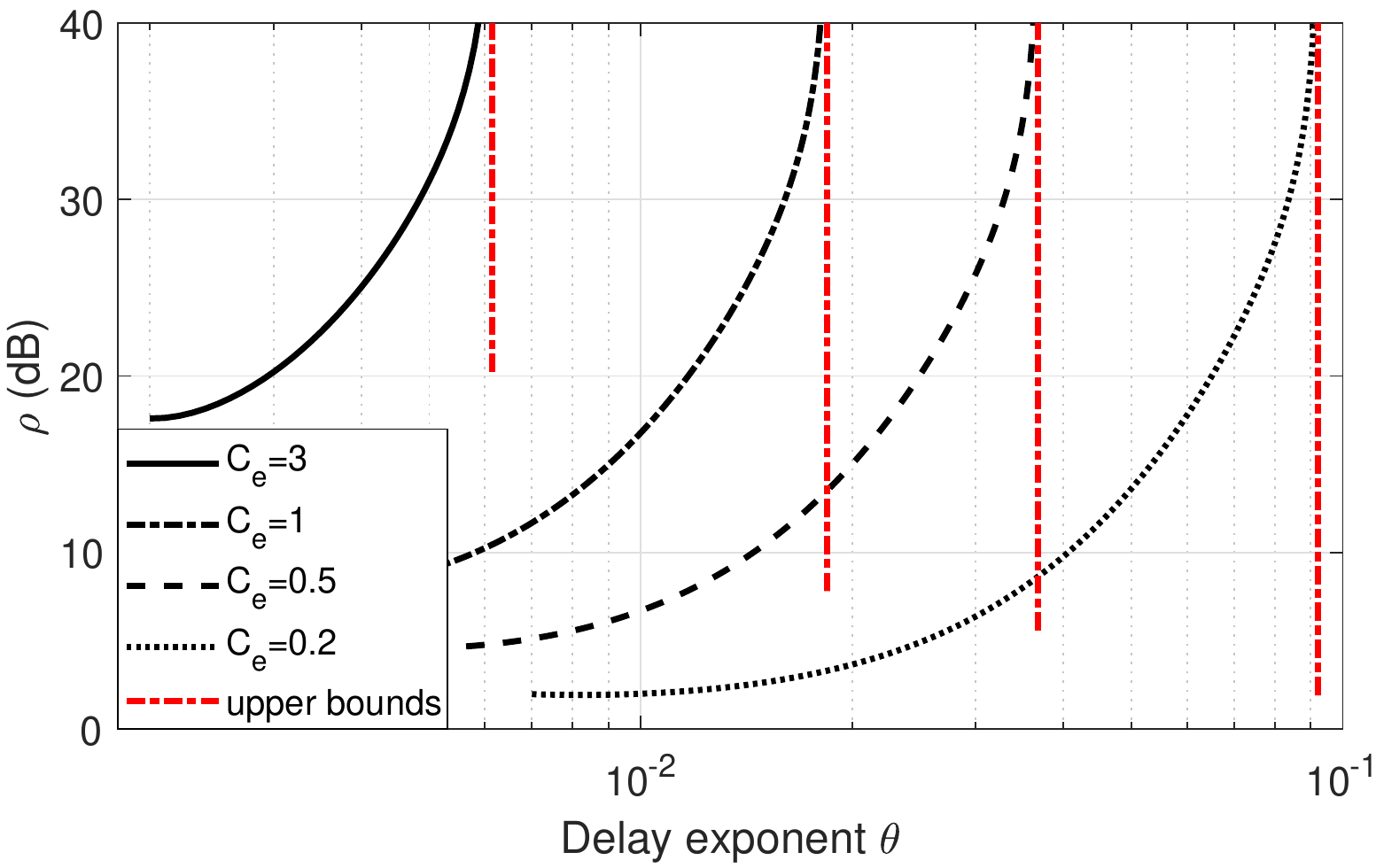}
	\end{center}
	\vspace{-2mm}
	\caption{Power delay profile for fixed EC buffers with $n=500$ and $\epsilon=10^{-4}$.}
	\label{power_delay_profile}
	\vspace{-2mm}
\end{figure}

\vspace{-2mm}
\begin{figure}[ht] 
	\begin{center}
		\includegraphics*[width=1\columnwidth]{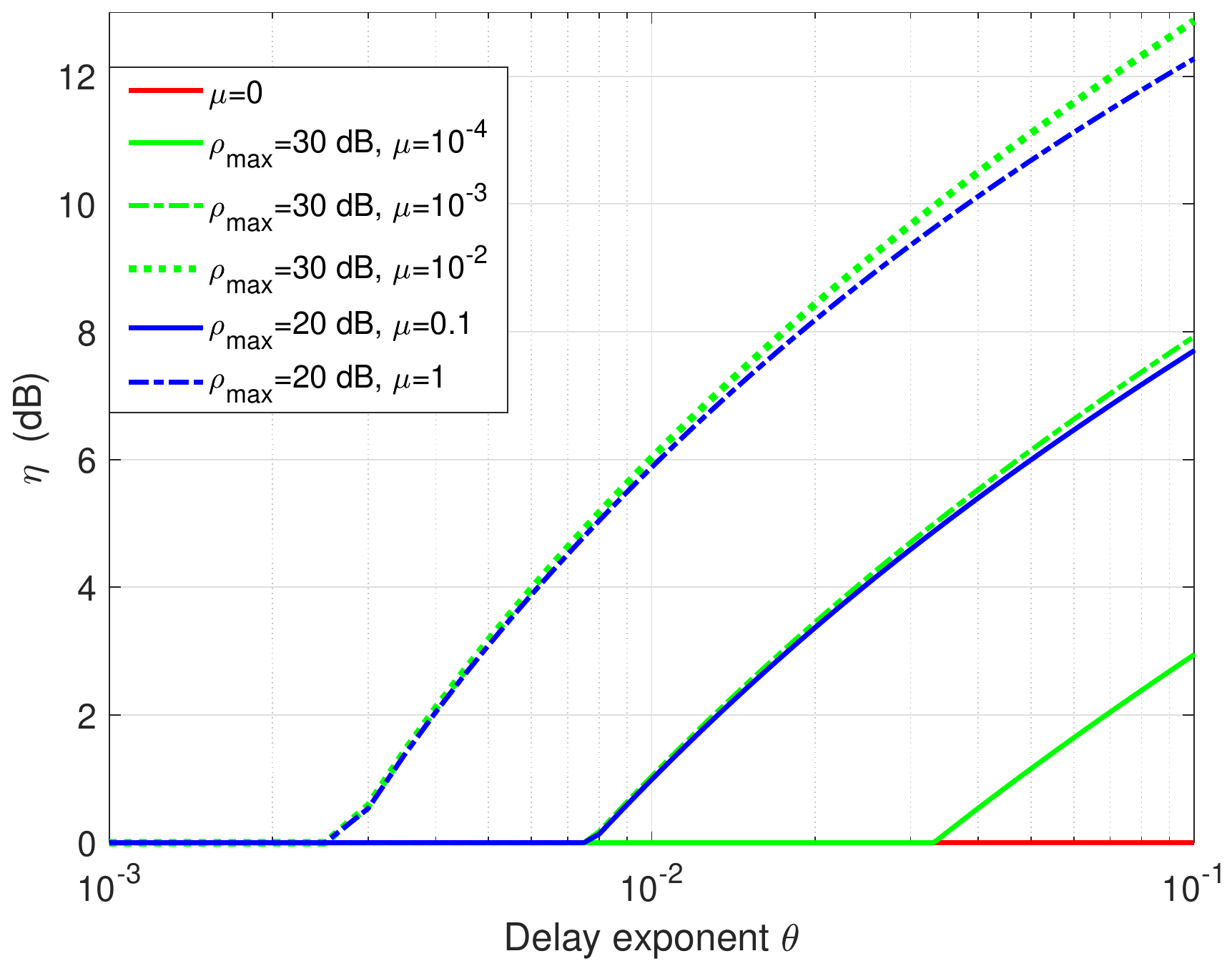}
	\end{center}
	\vspace{-2mm}
	\caption{Power saving vs the delay exponent $\theta$ for $n=500, \epsilon=10^{-4}$, and $\rho_{max}=20, 30$  dB.}
	\label{power_gain}
	\vspace{-2mm}
\end{figure} \vspace{-2mm}
\begin{figure}[ht] 
	\begin{center}
		\includegraphics*[width=1\columnwidth]{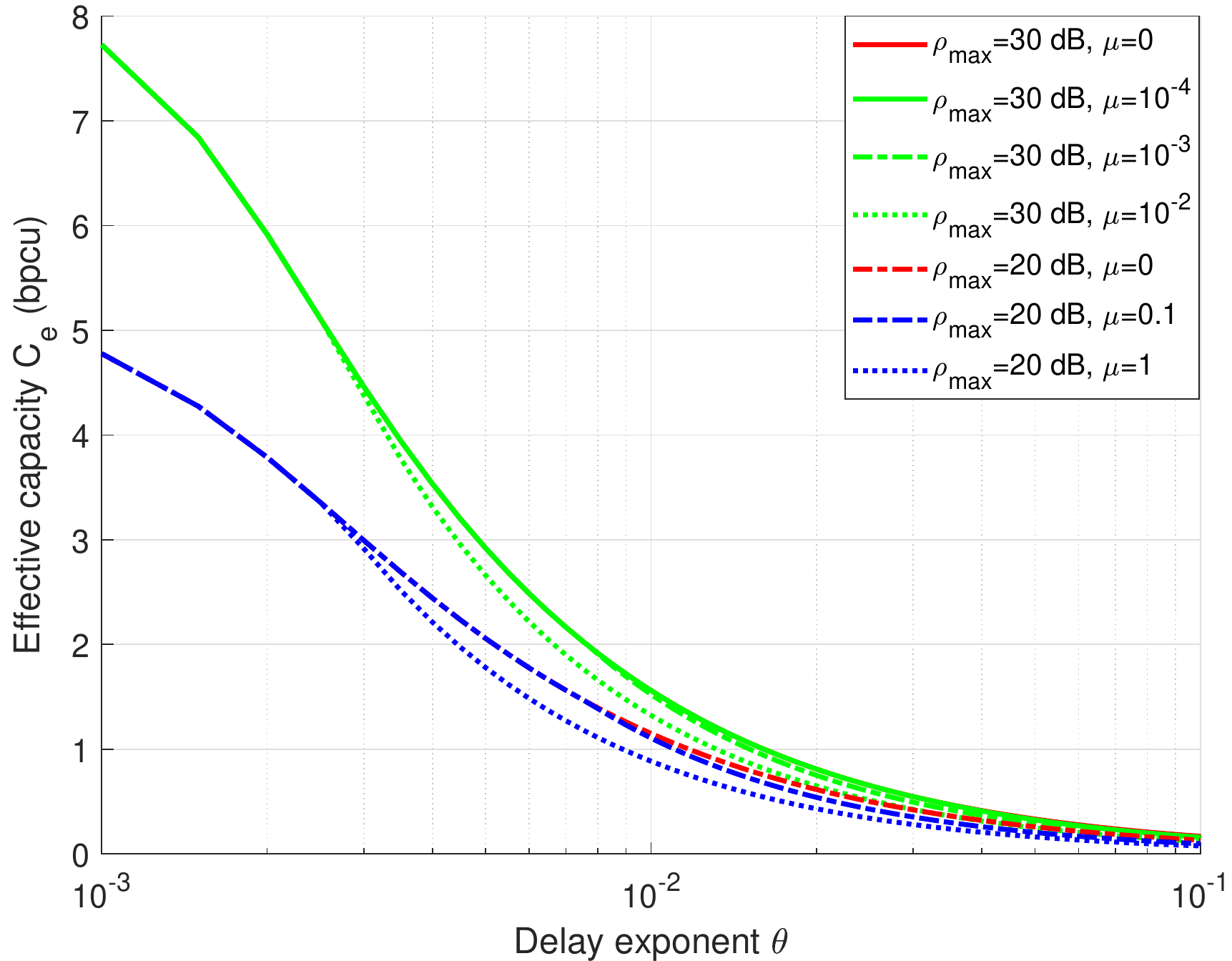}
	\end{center}
	\vspace{-2mm}
	\caption{Effect of power saving in quasi-static Rayleigh fading, where $n=500, \epsilon=10^{-4}$, and $\rho_{max}=20, 30$  dB.}
	\label{EC_difference}
	\vspace{-2mm}
\end{figure}
\vspace{0mm}

\section{Conclusion} \label{conclusion}
In this work, we thoroughly analyzed the effective capacity for short packets transmission in MTC networks. For quasi-static Rayleigh fading channels, we obtained a tight closed form approximation for the EC in terms of well known mathematical functions and characterized the optimum error probability for maximizing the EC. Furthermore, in contrary to the infinity blocklength case, we showed that in the finite blocklength regime, the EC is upper bounded in high SNR, which is a consequence of the short blocklength and error probability imposed by the system design. Finally, we discussed the power-delay relation and proposed an optimum power allocation scheme in the high SNR regime exploiting the asymptotic behavior of EC. The proposed power allocation strategy leads to considerable power saving at a very low loss in EC.

\vspace{-2mm}
\appendices 
\section{Proof Of Lemma \ref{lemma 1}}
Let us first define 
	\begin{align}\label{EC2}
	\psi(\rho,\theta,\epsilon)&= \mathbb{E}_{Z}\left[\epsilon+(1-\epsilon)e^{-n\theta r}\right] \\
	&=\int_{0}^{\infty}
	\left( \epsilon+(1-\epsilon)e^{-\theta n r}\right)  e^{-z} dz. 	
	\end{align}
	From \eqref{eq3} and \cite{eucnc}, we have		
	\begin{align}\label{e1}
	e^{-\theta n r}=(1+\rho z)^{\alpha} e^{\beta \gamma}, 
	\end{align}	
	where $\gamma=\sqrt{(1-(1+\rho z)^{-2})}$. Since $\gamma$ tends to 1 as $\rho \rightarrow \infty$ and $\beta \rightarrow 0$ for the cases under interest (namely, ultra reliable scenarios with finite block length, where $n$ is small and $\epsilon$ is small but not zero), we can then truncate the Maclaurin series as $e^{\beta \gamma} = 1+(\beta \gamma)+\frac{(\beta \gamma)^2}{2}$ and replace it into (\ref{e1}), then \eqref{EC2} becomes
	\begin{align}\label{general2}
	\begin{split}
	&\psi(\rho,\theta,\epsilon)=  \epsilon 
	+(1-\epsilon)\left[ \int_{0}^{\infty}
	(1+\rho z)^{\alpha}e^{-z} dz + \right. \\
	&\left. \beta\int_{0}^{\infty}\!
	(1+\rho z)^{\alpha} \gamma e^{-z} dz + \frac{\beta^2}{2}\int_{0}^{\infty}
	(1+\rho z)^{\alpha} \gamma^2 e^{-z} dz  \right].  
	\end{split}
	\end{align}
	The first integral reduces to $e^{\frac{1}{\rho}}  \rho^\alpha \Gamma\left(\alpha+1,\frac{1}{\rho} \right)$. After we apply Laurent's expansion for $\gamma$ \cite{Complex_Analysis}, we attain $\gamma\approx1-\tfrac{1}{2} \left(1+\rho z \right)^{-2}$. Then, the second and third integrals can be written as  
	$e^{\frac{1}{\rho}} \beta \rho^\alpha  \left( \Gamma\left(\alpha+1,\tfrac{1}{\rho} \right)-\tfrac{1}{2\rho^{2}}\Gamma\left(\alpha-1,\frac{1}{\rho}\right)\right) $, and $e^{\frac{1}{\rho}} \frac{\beta^2}{2} \rho^\alpha  \left( \Gamma\left(\alpha+1,\frac{1}{\rho} \right)-\tfrac{1}{\rho^{2}} \Gamma\left(\alpha-1,\frac{1}{\rho}\right) \right) $, respectively. After some algebraic manipulations we obtain $\psi=\epsilon+(1-\epsilon)\mathcal{J}$ where $\mathcal{J}$ is given by (\ref{c2.2}). \vspace{-2mm}
\section{Proof Of Lemma \ref{lemma 3}}
According to \cite{wolf}, the gamma function can be represented in terms of the generalized exponential integral $E_{1-a}(\cdot)$ as
\begin{align}\label{pd-0}
\Gamma\left( a,\tfrac{1}{\rho}\right) = \rho^{-a} E_{1-a}\left( \tfrac{1}{\rho}\right) \approx-\frac{\rho^{-a}}{a},
\end{align}
where $\lim\limits_{\rho\rightarrow \infty}E_{1-a}\left( \frac{1}{\rho}\right)\approx-\frac{1}{a}$ \cite[\S 8.19.6]{DLMF}.
For $a=\alpha+1$ and $a=\alpha-1$, it is observed that the approximation in (\ref{pd-0}) holds for $\{a \in \mathbb{R} | a < 1 \}$ which corresponds to all practical values of $\alpha$. Thus, by applying (\ref{pd-0}) into \eqref{c2.2}, we obtain 


\vspace{-0mm}
\begin{align}\label{c2.6}
\begin{split}
\mathcal{J_ \infty}\approx e^{\frac{1}{\rho}} \frac{1}{\rho}\left[-\frac{\kappa+1}{\alpha+1}+\frac{\kappa-\frac{\beta}{2}}{\alpha-1}\right]=e^{\frac{1}{\rho}} \frac{1}{\rho}\mathcal{F}.
\end{split}
\end{align}
From (\ref{Rayleigh}), we isolate $\mathcal{J}$ as
\vspace{-0mm}
\begin{align}\label{pd-1}
\begin{split}
\mathcal{J_\infty}=\frac{e^{-n \theta \ C_e}-\epsilon}{1-\epsilon}.
\end{split}
\end{align}
After manipulating (\ref{c2.6}) and (\ref{pd-1}) we get (\ref{pd0}) with the help of the definition of the Lambert-W function which is defined as the solution to \cite{Corless}
	\begin{align}\label{Lambert}
	f(x)=x e^x, \ where \ \
	x=f^{-1}(x)=\mathcal{W}(x e^x)
	\end{align}  \vspace{-2mm}
\section{Proof Of Theorem \ref{theorem 2}}
As envisioned from Fig. 2, the effective capacity is an increasing function of the transmit SNR. However, at high SNR, the rate of increase of effective capacity decreases due to the asymptotic behaviour which is concluded in Remark 1. Thus, the optimal power allocation is achieved at equality of the constraint given by (10 c). From (\ref{Rayleigh}) and (\ref{c2.6}), we have
\begin{align}\label{pd8}
&\frac{\partial C_e}{\partial \rho}= 
-\frac{(1-\epsilon)}{n \theta \psi} \frac{\partial \mathcal{J}}{\partial \rho} =-\frac{(1-\epsilon)}{n \theta \psi}\mathcal{F}\left(\frac{-1}{\rho^2}e^{\frac{1}{\rho}}-\frac{1}{\rho^3} e^{\frac{1}{\rho}} \right)  \notag \\
&\!=\!\frac{(1-\epsilon)}{n \theta \psi}\frac{\mathcal{F}}{\rho^2}e^{\frac{1}{\rho}}\left(1+\frac{1}{\rho} \right)\stackrel{(l)}{\approx}\frac{(1-\epsilon)}{n \theta \psi}\frac{\mathcal{F}}{\rho^2}\left(1+\frac{1}{\rho} \right)^2\!, 
\end{align}
where in step $(l)$, we applied the first order Taylor expansion of $e^{\frac{1}{\rho}}\approx1+\frac{1}{\rho}$ around zero. According to (\ref{ub5}), as $\rho \rightarrow \infty$, $\psi$ converges to $\epsilon$. At equality of the constraint (\ref{c}), we attain
\begin{align}\label{pd9}
\begin{split}
&\frac{(1-\epsilon)}{n \theta \epsilon}\frac{\mathcal{F}}{\rho^2}\left(1\!+\!\frac{1}{\rho} \right)^2\!=\!\frac{(1-\epsilon)}{n \theta \epsilon}\mathcal{F}\left(\frac{1}{\rho^4}\!+\!\frac{2}{\rho^3}\!+\!\frac{1}{\rho^2} \right)\!=\!\mu,
\end{split}
\end{align}
which leads to 
\vspace{-2mm}
\begin{align}\label{pd10}
\begin{split}
\left(\frac{1}{\rho^4}+\frac{2}{\rho^3}+\frac{1}{\rho^2} \right)=\frac{\mu n \theta \epsilon}{(1-\epsilon)\mathcal{F}}.
\end{split}
\end{align}
Leveraging the constraint (\ref{b}) and (\ref{pd10}), we obtain the positive solution to this problem as in Theorem \ref{theorem 2}.
\vspace{-2mm}

\bibliographystyle{IEEEtran}
\bibliography{di}

\begin{thebibliography}{10}
\providecommand{\url}[1]{#1}
\csname url@samestyle\endcsname
\providecommand{\newblock}{\relax}
\providecommand{\bibinfo}[2]{#2}
\providecommand{\BIBentrySTDinterwordspacing}{\spaceskip=0pt\relax}
\providecommand{\BIBentryALTinterwordstretchfactor}{4}
\providecommand{\BIBentryALTinterwordspacing}{\spaceskip=\fontdimen2\font plus
\BIBentryALTinterwordstretchfactor\fontdimen3\font minus
  \fontdimen4\font\relax}
\providecommand{\BIBforeignlanguage}[2]{{%
\expandafter\ifx\csname l@#1\endcsname\relax
\typeout{** WARNING: IEEEtran.bst: No hyphenation pattern has been}%
\typeout{** loaded for the language `#1'. Using the pattern for}%
\typeout{** the default language instead.}%
\else
\language=\csname l@#1\endcsname
\fi
#2}}
\providecommand{\BIBdecl}{\relax}
\BIBdecl

\bibitem{paper1}
P.~Popovski, ``{Ultra-reliable communication in 5G wireless systems},'' in
  \emph{1st International Conference on 5G for Ubiquitous Connectivity}, Nov
  2014, pp. 146--151.

\bibitem{paper2}
G.~Durisi, T.~Koch, and P.~Popovski, ``{Toward Massive, Ultrareliable, and
  Low-Latency Wireless Communication with Short Packets},'' \emph{Proceedings
  of the IEEE}, vol. 104, no.~9, pp. 1711--1726, 2016.

\bibitem{TVT3}
J.~Huang, C.~C. Xing, Y.~Qian, and Z.~J. Haas, ``Resource allocation for
  multicell device-to-device communications underlaying 5g networks: A
  game-theoretic mechanism with incomplete information,'' \emph{IEEE
  Transactions on Vehicular Technology}, vol.~67, no.~3, pp. 2557--2570, March
  2018.

\bibitem{Yang_J}
G.~Durisi, T.~Koch, J.~Ostman, Y.~Polyanskiy, and W.~Yang, ``Short-packet
  communications over multiple-antenna rayleigh-fading channels,'' \emph{IEEE
  Transactions on Communications}, vol.~64, no.~2, pp. 618--629, Feb 2016.

\bibitem{paper5}
M.~Gursoy, ``Throughput analysis of buffer-constrained wireless systems in the
  finite blocklength regime,'' in \emph{EURASIP Journal on Wireless
  Communucations and Networking 2013}, 2013.

\bibitem{Yang2014c}
W.~Yang, G.~Durisi, T.~Koch, and Y.~Polyanskiy, ``{Quasi-Static
  Multiple-Antenna Fading Channels at Finite Blocklength},'' \emph{IEEE Trans.
  Inf. Theory}, vol.~60, no.~7, pp. 4232--4265, jul 2014.

\bibitem{TVT1}
Y.~Hu, M.~Serror, K.~Wehrle, and J.~Gross, ``Finite blocklength performance of
  cooperative multi-terminal wireless industrial networks,'' \emph{IEEE
  Transactions on Vehicular Technology}, pp. 1--1, 2018.

\bibitem{TVT2}
Y.~Hu, J.~Gross, and A.~Schmeink, ``On the capacity of relaying with finite
  blocklength,'' \emph{IEEE Transactions on Vehicular Technology}, vol.~65,
  no.~3, pp. 1790--1794, March 2016.

\bibitem{Polyanskiy2010}
Y.~Polyanskiy, H.~V. Poor, and S.~Verdu, ``{Channel Coding Rate in the Finite
  Blocklength Regime},'' \emph{IEEE Trans. Inf. Theory}, vol.~56, no.~5, pp.
  2307--2359, may 2010.

\bibitem{paper6}
D.~Wu and R.~Negi, ``{Effective capacity: a wireless link model for support of
  quality of service},'' \emph{IEEE Trans. Wireless Commun}, vol.~2, no.~4, pp.
  630--643, 2003.

\bibitem{paper9}
L.~Musavian and T.~Le-Ngoc, ``{QoS-based power allocation for cognitive radios
  with AMC and ARQ in Nakagami-m fading Channels},'' in \emph{Trans. Emerging
  Tel. Tech.}, vol.~27, 2014, pp. 266--277.

\bibitem{paper8}
L.~Musavian and Q.~Ni, ``Effective capacity maximization with statistical delay
  and effective energy efficiency requirements,'' \emph{IEEE Transactions on
  Wireless Communications}, vol.~14, no.~7, pp. 3824--3835, July 2015.

\bibitem{eucnc}
M.~Shehab, E.~Dosti, H.~Alves, and M.~Latva-aho, ``On the effective capacity of
  {MTC} networks in the finite blocklength regime,'' in \emph{EUCNC 2017},
  Oulu, Finland, Jun. 2017.

\bibitem{Gradshteyn}
I.~S. Gradshteyn and I.~M. Ryshik, \emph{{Table of Integrals, Series, and
  Products}}, 7th~ed.\hskip 1em plus 0.5em minus 0.4em\relax London: Academic
  Press, 2007.

\bibitem{Corless}
R.~Corless and G.~Gonnet~et al., ``{On the Lambert W Function},''
  \emph{Advances in Computational Mathematics}, vol.~5, pp. 329--359, 2007.

\bibitem{Complex_Analysis}
K.~I. R.~Rodriguez and J.~Gilman, \emph{Complex Analysis: In the Spirit of
  Lipman Bers.}\hskip 1em plus 0.5em minus 0.4em\relax Springer, 2012.

\bibitem{wolf}
\BIBentryALTinterwordspacing
Wolfram|Alpha, access July 14, 2017. [Online]. Available:
  \url{http://functions.wolfram.com/06.06.27.0003.01}
\BIBentrySTDinterwordspacing

\bibitem{DLMF}
\BIBentryALTinterwordspacing
\emph{{\it NIST Digital Library of Mathematical Functions}}, ~F. Olver, A.~B.
  {Olde Daalhuis}, et al. [Online]. Available: \url{http://dlmf.nist.gov/}
\BIBentrySTDinterwordspacing

\end{thebibliography}
\end{document}